\documentclass[12pt,a4paper]{article}

\usepackage{amsmath,amssymb,amsthm}
\usepackage[margin=2.5cm]{geometry}
\usepackage[colorlinks=true,linkcolor=blue,citecolor=blue,urlcolor=blue]{hyperref}
\usepackage{enumitem}
\usepackage{booktabs}
\usepackage{graphicx}
\usepackage{float}

\theoremstyle{plain}
\newtheorem{theorem}{Theorem}[section]
\newtheorem{proposition}[theorem]{Proposition}

\theoremstyle{definition}

\theoremstyle{remark}
\newtheorem{remark}[theorem]{Remark}

\newcommand{\R}{\mathbb{R}}
\newcommand{\diag}{\operatorname{diag}}

\newcommand{\nhat}{\hat{\mathbf{n}}}

\title{A Metric for Three-Dimensional Color Discrimination\\
Derived from V1 Population Fisher Information}

\author{Michael Menke}
\date{}

\begin{document}
\maketitle

\begin{abstract}
\noindent
We derive a Riemannian metric on three-dimensional color space from
the Fisher information of neural population codes in the visual pathway.  Photoreceptor adaptation, retinal opponent channels, and cortical population encoding each map onto a geometric construction, producing
a metric tensor whose components correspond to measurable neural
quantities.  The resulting 17-parameter model is fitted jointly to
four independent threshold datasets: MacAdam's~(1942) chromaticity
ellipses, the Koenderink et~al.\ (2026) three-dimensional
ellipsoids, Wright's~(1941) wavelength discrimination function,
and the Huang et~al.\ (2012) threshold color difference ellipses,
covering 96 independently measured discrimination conditions across
varied chromaticities and luminances. The joint fit achieves STRESS of 23.9 on MacAdam, 20.8 on Koenderink et al., 30.1 on Wright, and 30.8 on Huang et al.  
\end{abstract}

\section{Introduction}\label{sec:intro}

Defining a perceptually meaningful metric on color space has been a
central problem in color science since Helmholtz~\cite{helmholtz91}
proposed a Riemannian metric based on Weber's law.  Subsequent
proposals by Schr\"odinger~\cite{schrodinger20} and
Stiles~\cite{stiles46} refined the metric but could not
accommodate later empirical data.  MacAdam's~\cite{macadam42}
measurements of chromatic discrimination ellipses shifted the field
from theoretical metrics to empirical characterization,
establishing that perceptual color space is locally anisotropic and
spatially inhomogeneous.  Modern color difference formulas such as
CIEDE2000~\cite{luo01} represent decades of empirical refinement but
remain parametric corrections layered onto a pre-existing coordinate
system.
\\
\\
\noindent
Here we take a different approach: we derive the discrimination metric
from the Fisher information of a population of chromatically tuned V1
neurons.  Each component of the metric tensor traces to a specific
neural population with a specific noise model and transduction
function, and each free parameter corresponds to a quantity
constrained by neurophysiology.  Two cross-channel normalization
mechanisms, koniocellular suppression by L-M opponent activity and
S-cone disinhibition of the non-cardinal population, provide
chromaticity-dependent gain control that resolves tradeoffs between
near-white and spectrum-locus chromaticities.  The model is validated
against four independent psychophysical datasets spanning chromaticity
ellipses, three-dimensional ellipsoids, monochromatic wavelength
discrimination, and threshold colour differences at varied
chromaticities and luminances.

\section{Photoreceptor Adaptation}\label{sec:adaptation}

\subsection{Cone excitations}

A color stimulus with CIE 1931 tristimulus values $(X,Y,Z)$ produces
cone excitations via the Smith--Pokorny~\cite{smith75}
transformation, based on the cone spectral sensitivities of
Bowmaker \& Dartnall~\cite{bowmaker80}:
\begin{equation}\label{eq:cone}
  \begin{pmatrix} L \\ M \\ S \end{pmatrix}
  = M_{\mathrm{SP}}
  \begin{pmatrix} X \\ Y \\ Z \end{pmatrix},
  \qquad
  M_{\mathrm{SP}} =
  \begin{pmatrix}
    0.15514 & \phantom{-}0.54312 & -0.03286 \\
   -0.15514 & \phantom{-}0.45684 & \phantom{-}0.03286 \\
    0 & \phantom{-}0 & \phantom{-}0.01608
  \end{pmatrix}.
\end{equation}

\subsection{Von Kries adaptation}

Chromatic adaptation is modeled by independent gain control on each
cone class~\cite{vonkries05}:
\begin{equation}\label{eq:vonkries}
  L_a = L/L_W, \qquad M_a = M/M_W, \qquad S_a = S/S_W,
\end{equation}
where $(L_W, M_W, S_W)$ are the cone excitations of the adaptation
illuminant at a reference luminance $Y_{\mathrm{ref}}$.  This
multiplicative normalization is supported by the physiology of cone
photoreceptor gain control~\cite{schneeweis99} and by psychophysical
evidence for the approximate validity of von~Kries
adaptation~\cite{brainard92}.  The adapted signals are dimensionless
ratios centered at unity for the adaptation white:
$L_a(\text{white}) = M_a(\text{white}) = S_a(\text{white}) = 1$.

\section{Opponent Channel Formation}\label{sec:opponent}

\subsection{Parvocellular coordinate}

The midget ganglion cells of the primate retina compute L versus M
opponency~\cite{dacey00,kaplan86}.  Weber's law ($\Delta L / L
\approx \text{const}$; \cite{hood86}) implies that the coordinate
rendering equal stimulus steps into equal perceptual distances is the
logarithm of the L/M ratio.  However, midget cells also exhibit
contrast gain control~\cite{kaplan86}, and their response saturates
at high chromatic contrasts.  To capture both the near-neutral Weber
regime and the saturation at extreme chromaticities, we pass the
log-ratio through a Naka--Rushton~\cite{naka66} compressive
nonlinearity:
\begin{equation}\label{eq:z}
  u = \log(L_a / M_a), \qquad
  z = z_{\max}\;\mathrm{sign}(u)\;
  \frac{|u|^{p_z}}{|u|^{p_z} + \kappa_z},
\end{equation}
where $z_{\max}$ is the saturation ceiling, $p_z$ controls the shape
of the compression, and $\kappa_z$ is the half-saturation constant.
Near the adaptation white ($|u| \ll \kappa_z^{1/p_z}$), the
coordinate reduces to $z \approx (z_{\max} / \kappa_z)\, u$,
recovering Weber's law with effective gain
$z_{\max} / \kappa_z$.  At extreme chromaticities
($|u| \gg \kappa_z^{1/p_z}$), $z$ saturates toward $\pm z_{\max}$.
This parallels the Naka--Rushton nonlinearity applied to the
koniocellular channel (Section~\ref{sec:konio})
and treats both opponent pathways symmetrically.

\begin{proposition}\label{prop:invariance}
The coordinate $z$ depends only on chromaticity, not on luminance.
That is, changing $Y$ at fixed $(x,y)$ leaves $z$ unchanged.
\end{proposition}

\begin{proof}
At fixed chromaticity $(x,y)$, the tristimulus values
$(X,Y,Z)$ are proportional to $Y$:
$X = Yx/y$, $Z = Y(1-x-y)/y$.  Since $M_{\mathrm{SP}}$ is linear,
all cone signals scale by the same factor $Y'/Y$ under a luminance
change.  This common factor cancels in the ratio $L_a/M_a$, leaving
$z$ unchanged.  The same argument applies to $r = S_a/(L_a + M_a)$
and therefore to $w$.
\end{proof}

\subsection{Koniocellular coordinate}\label{sec:konio}

The small bistratified ganglion cells compute S versus $(L+M)$
opponency~\cite{dacey94}.  The relevant signal is the ratio
\begin{equation}\label{eq:r}
  r = S_a / (L_a + M_a),
\end{equation}
which is also luminance-invariant by the same argument.  At the adaptation
white, $r = 1/(1+1) = 0.5$.

This ratio is passed through a Naka--Rushton~\cite{naka66}
compressive nonlinearity:
\begin{equation}\label{eq:w}
  w = \frac{r^{p_s}}{r^{p_s} + \kappa_s},
\end{equation}
modeling the saturation of the koniocellular pathway's
contrast-response function~\cite{tailby08}.  The function
$w\colon\R^+ \to (0,1)$ is monotonically increasing, with
$w \to 0$ as $r \to 0$ and $w \to 1$ as $r \to \infty$.

\begin{remark}
The parvocellular channel compresses $\log(L_a/M_a)$; the koniocellular
channel compresses $S_a/(L_a + M_a)$.  The two pathways retain
different exponents and half-saturation constants, reflecting the
genuine physiological differences between them: the parvocellular
pathway (${\sim}80\%$ of LGN neurons) has higher contrast gain and
broader dynamic range~\cite{kaplan86}, while the koniocellular pathway
(${\sim}10\%$ of LGN neurons) has narrower dynamic range and stronger
compressive nonlinearity~\cite{hendry00,tailby08}. 
\end{remark}

\section{Cortical Population Fisher Information}\label{sec:fisher}

\subsection{Fisher information for a single Poisson neuron}
 
The link between neural encoding and Riemannian geometry is Fisher
information.  Consider a single neuron whose spike count $r$ on a
given trial is drawn from a Poisson distribution with mean firing
rate $f(\mathbf{s})$, where $\mathbf{s} = (s_1, s_2)$ is a
two-dimensional stimulus (here, the opponent coordinates $(z,w)$).
The Poisson probability of observing $r$ spikes is
$P(r) = f^r\,e^{-f}/r!$, so the log-likelihood is
\[
  \ell = r \log f(\mathbf{s}) - f(\mathbf{s}) - \log(r!).
\]
The Fisher information matrix is defined as the expected negative
Hessian of the log-likelihood:
\[
  I_{ij}(\mathbf{s})
  = -\mathbb{E}\!\left[\frac{\partial^2 \ell}
  {\partial s_i\,\partial s_j}\right].
\]
To evaluate this for the Poisson case, we compute the derivatives
of $\ell$.  The first partial derivative is
\[
  \frac{\partial \ell}{\partial s_i}
  = \frac{r}{f}\frac{\partial f}{\partial s_i}
  - \frac{\partial f}{\partial s_i}
  = \left(\frac{r}{f} - 1\right)\frac{\partial f}{\partial s_i}.
\]
Differentiating again,
\[
  \frac{\partial^2 \ell}{\partial s_i\,\partial s_j}
  = -\frac{r}{f^2}
  \frac{\partial f}{\partial s_i}\frac{\partial f}{\partial s_j}
  + \left(\frac{r}{f} - 1\right)
  \frac{\partial^2 f}{\partial s_i\,\partial s_j}.
\]
Taking the expectation with $\mathbb{E}[r] = f$, the second term
vanishes (since $\mathbb{E}[r/f - 1] = 0$), leaving
\begin{equation}\label{eq:fisher_general}
  I_{ij}(\mathbf{s})
  = -\mathbb{E}\!\left[\frac{\partial^2 \ell}
  {\partial s_i\,\partial s_j}\right]
  = \frac{1}{f(\mathbf{s})}
  \frac{\partial f}{\partial s_i}
  \frac{\partial f}{\partial s_j}.
\end{equation}
In matrix form, $I = f^{-1}\,\nabla f\,\nabla f^T$, where
$\nabla f = (\partial f/\partial s_1,\;\partial f/\partial s_2)^T$.
The Poisson Fisher information depends only on the tuning curve
$f(\mathbf{s})$ and its gradient; higher-order properties of the
spike count distribution do not enter.

\subsection{From Fisher information to discrimination ellipsoids}
 
The Fisher information matrix is the natural metric tensor for
discrimination because it determines the smallest detectable stimulus
change.  To see why, consider an observer who must estimate the
stimulus $\mathbf{s}$ from the spike counts of a neural population.
Any estimator $\hat{\mathbf{s}}$ will fluctuate from trial to trial
due to the randomness of spiking.  The covariance matrix
$\mathrm{Cov}[\hat{\mathbf{s}}]$ describes the shape and size of
these fluctuations: its diagonal entries are the variances of each
coordinate estimate, its off-diagonal entries capture correlations
between estimation errors, and its eigenvectors define the directions
in stimulus space along which the estimate is most and least precise.
\\
\\
The Cram\'er--Rao inequality~\cite{cover06} places a lower bound on
these fluctuations.  For any unbiased estimator,
\begin{equation}\label{eq:cramer_rao}
  \mathrm{Cov}[\hat{\mathbf{s}}]
  \;\succeq\; I(\mathbf{s})^{-1},
\end{equation}
where $\succeq$ denotes the positive semi-definite ordering, meaning
$\mathrm{Cov} - I^{-1}$ has no negative eigenvalues.  In words: no
estimation strategy can reduce the trial-to-trial variability below
$I^{-1}$ in any direction.  Directions in stimulus space where the
Fisher information has large eigenvalues
permit precise estimation; directions where the eigenvalues are small
produce large estimation noise.
\\
\\
Now consider an ideal observer who must decide whether the stimulus
is $\mathbf{s}$ (no change) or $\mathbf{s} + \Delta\mathbf{s}$
(a small shift).  Under the null hypothesis the stimulus estimate
is distributed approximately as
$\hat{\mathbf{s}} \sim \mathcal{N}(\mathbf{s},\; I^{-1})$, and
under the alternative as
$\hat{\mathbf{s}} \sim \mathcal{N}(\mathbf{s} + \Delta\mathbf{s},\;
I^{-1})$.  These are two multivariate Gaussians with the same
covariance but shifted means.  The optimal discriminability between
them is the Mahalanobis distance of the mean shift measured in the
inverse covariance of the noise:
\begin{equation}\label{eq:dprime}
  (d')^2 = \Delta\mathbf{s}^T\,
  \bigl(I^{-1}\bigr)^{-1}\, \Delta\mathbf{s}
  = \Delta\mathbf{s}^T\, I(\mathbf{s})\, \Delta\mathbf{s}.
\end{equation}
This is a standard result in signal detection
theory~\cite{dayan01}: for two Gaussians with shared covariance
$\Sigma$, the detection sensitivity is
$d'^2 = \Delta\boldsymbol{\mu}^T \Sigma^{-1} \Delta\boldsymbol{\mu}$,
and here $\Sigma = I^{-1}$ so $\Sigma^{-1} = I$.
\\
\\
Threshold corresponds to a fixed criterion
$d' = c$, giving the set of just-detectable shifts as
\[
  \bigl\{\Delta\mathbf{s}
  : \Delta\mathbf{s}^T\, I\, \Delta\mathbf{s} = c^2 \bigr\},
\]
which is an ellipsoid in stimulus space with shape matrix
$\Sigma = c^2\, I^{-1}$.  The semi-axes of this ellipsoid are
$c/\sqrt{\mu_k}$, where $\mu_k$ are the eigenvalues of $I$.
\\
\\
This is the geometric meaning of setting $g = I$: the Fisher
information matrix defines a Riemannian metric, and discrimination
thresholds trace out the unit ball of the inverse metric.  The
criterion constant $c$ is a global scale factor that depends on the
experimental procedure (e.g., the target percent correct in a
forced-choice task) and is absorbed by the STRESS optimal scale
factor $F$ when comparing with data.

\subsection{Directional tuning and the rank-1 structure}
 
V1 neurons respond selectively to chromatic direction.  A neuron
tuned to a preferred direction $\nhat = (\cos\theta,\,\sin\theta)$
in the $(z,w)$ plane has a firing rate that depends on the stimulus
only through the scalar projection onto $\nhat$:
\begin{equation}\label{eq:tuning}
  f(\mathbf{s}) = \phi(\nhat \cdot \mathbf{s}),
\end{equation}
where $\phi$ is the neuron's (generally nonlinear) contrast-response
function.  Taking partial derivatives by the chain rule,
\[
  \frac{\partial f}{\partial s_i}
  = \phi'(\nhat \cdot \mathbf{s})\;\hat{n}_i,
\]
so the gradient $\nabla f = \phi'\,\nhat$.  Substituting
into~\eqref{eq:fisher_general} gives
\begin{equation}\label{eq:fisher_rank1}
  I = \frac{[\phi']^2}{\phi}\;\nhat\,\nhat^T.
\end{equation}
Note that any single Poisson neuron has rank-1 Fisher information,
since $I = f^{-1} \nabla f\,\nabla f^T$ is always an outer product
regardless of tuning shape.  The directional tuning
assumption~\eqref{eq:tuning} determines the \emph{direction} of this
rank-1 contribution: it is aligned with the neuron's preferred
chromatic direction $\nhat$.  The scalar prefactor
$[\phi']^2/\phi$ depends on the neuron's gain and operating point
but the geometric structure (the outer product $\nhat\,\nhat^T$)
depends only on the preferred direction.
 
\subsection{Population additivity}
 
Fisher information is additive across conditionally independent
neurons (independent spike counts given the
stimulus)~\cite{cover06}.  Cortical neurons do exhibit noise
correlations, but for populations with shared tuning the dominant
effect of correlations is to rescale the aggregate Fisher information
rather than change its directional structure~\cite{dayan01}.
If a population of $N_k$ neurons shares
preferred direction $\nhat_k$ and has aggregate scalar Fisher
information $\lambda_k$ (absorbing population size, per-neuron gain,
and operating point), then the total Fisher information of the
population is
\begin{equation}\label{eq:pop_sum}
  I_{\text{pop}} = \sum_k \lambda_k\;\nhat_k\,\nhat_k^T.
\end{equation}
For a metric tensor, the key requirement is that this matrix be
positive definite, which holds whenever the set of preferred
directions $\{\nhat_k\}$ spans $\mathbb{R}^2$ (i.e., not all neurons
are tuned to the same direction).
 
\subsection{Three V1 sub-populations}\label{sec:three_pops}
 
Electrophysiological recordings show that V1 neurons span a
continuum of chromatic preferences, but with clustering near the
cardinal axes and a broad distribution of intermediate (non-cardinal)
tunings~\cite{lennie90,johnson01,johnson04}.  We model this
distribution as three discrete sub-populations:
 
\begin{enumerate}[label=(\roman*)]
\item \emph{Parvocellular-recipient neurons}, tuned along the L-M
  axis ($\nhat = \mathbf{e}_z$).  Their aggregate Fisher information
  is $\lambda_P\,\mathbf{e}_z\mathbf{e}_z^T$.
\item \emph{Koniocellular-recipient neurons}, tuned along the S-cone
  axis ($\nhat = \mathbf{e}_w$).  Their aggregate Fisher information
  is $\lambda_K\,\mathbf{e}_w\mathbf{e}_w^T$, where $\lambda_K$
  differs from $\lambda_P$ because the koniocellular pathway has
  fewer neurons and different contrast gain~\cite{hendry00,tailby08}.
  We write $\lambda_K = \gamma^2\,\mathcal{N}(z)$, where $\gamma$
  captures the baseline sensitivity ratio and $\mathcal{N}(z)$ is a
  cross-channel normalization factor (see \ref{sec:ccn}).
\item \emph{Non-cardinal neurons}, tuned to a chromaticity-dependent
  intermediate direction $\nhat(\theta)$ where
  $\theta = \theta(z,w)$.  Their aggregate Fisher information is
  $\lambda_N\,\nhat\,\nhat^T$, where
  $\lambda_N = \eta^2\,\mathcal{S}(w)$ incorporates both the
  population weight $\eta^2$ and an S-cone disinhibition factor
  $\mathcal{S}(w)$ (Section~\ref{sec:scone_disinh}).
\end{enumerate}
 
\noindent
Applying~\eqref{eq:pop_sum} with these three populations gives the
combined Fisher information at fixed luminance:
\begin{equation}\label{eq:G_fixed_Y}
  G_0(z,w) = \lambda_P\,\mathbf{e}_z\mathbf{e}_z^T
  + \gamma^2\mathcal{N}(z)\,\mathbf{e}_w\mathbf{e}_w^T
  + \eta^2\mathcal{S}(w)\,\nhat\,\nhat^T.
\end{equation}
Since $\mathbf{e}_z$ and $\mathbf{e}_w$ are orthogonal and both
carry strictly positive coefficients, the first two terms alone form
a positive definite diagonal matrix, and the non-cardinal term (which
is positive semi-definite) can only increase the eigenvalues.  The
matrix is therefore positive definite for all values of $\theta$ and
all $\eta \geq 0$.
 \\
 \\
Because only ratios of the population weights are observable in
threshold data (STRESS normalizes by an optimal global scale factor),
we set $\lambda_P = 1$ without loss of generality.  The parameters
$\gamma$ and $\eta$ then measure konio and non-cardinal sensitivity
relative to the parvocellular population.
 
\subsection{Luminance scaling}\label{sec:lum_scaling}
 
The aggregate scalar Fisher information $\lambda_k$ of each
population depends on luminance through the firing rates of its
neurons.  If firing rates scale as a power law of luminance,
$f(\mathbf{s}; Y) = c(\mathbf{s})\,Y^a$, then the Poisson Fisher
information~\eqref{eq:fisher_general} becomes
\begin{equation}\label{eq:poisson_lum}
  I_{ij}(\mathbf{s}; Y)
  = \frac{1}{c\,Y^a}
  \frac{\partial(c\,Y^a)}{\partial s_i}
  \frac{\partial(c\,Y^a)}{\partial s_j}
  = \frac{[c'_i][c'_j]}{c}\;Y^a,
\end{equation}
where $c'_i = \partial c/\partial s_i$.  The $Y^{2a}$ from the
numerator partially cancels the $Y^a$ in the denominator (the
Poisson variance), leaving a net scaling of $Y^a$.  When $a = 0$ the chromatic Fisher information is luminance-independent, meaning the von~Kries
adaptation in the earlier stages has fully normalized the luminance
dependence.
\\
\\
Different pathways have different luminance-response
exponents~\cite{kaplan86,irvin93}.  Writing $a_P$ for the parvo
exponent and $a_K$ for the konio exponent, the luminance-dependent
Fisher information of each population factorizes as
$\lambda_k(z,w) \cdot Y^{a_k}$.  We absorb the
chromaticity-dependent scalar $\lambda_k(z,w)$ into the population
parameters (with $\lambda_P = 1$ by the normalization convention
above), giving the full metric
\begin{equation}\label{eq:G}
  G(z,w;\,Y) = G_P + G_K + G_N,
\end{equation}
where
\begin{align}
  G_P &= Y^{a_P}\;\mathbf{e}_z\mathbf{e}_z^T,
  \label{eq:GP}\\
  G_K &= \mathcal{N}(z)\;\gamma^2\, Y^{a_K}\;
  \mathbf{e}_w\mathbf{e}_w^T,
  \label{eq:GK}\\
  G_N &= \eta^2\,\mathcal{S}(w)\, Y^{(a_P + a_K)/2}\;
  \nhat\,\nhat^T.
  \label{eq:GN}
\end{align}
The parvo and konio terms each carry their own luminance exponent.
The non-cardinal term uses $Y^{(a_P+a_K)/2}$, the geometric mean of
the two pathway exponents.  This is a modeling choice reflecting the
convergent input that non-cardinal V1 neurons receive from both
the parvocellular and koniocellular streams; we do not have direct
measurements of a separate non-cardinal luminance exponent, so we
treat it as intermediate between the two pathways it draws from.

\subsection{Cross-channel normalization}\label{sec:ccn}
 
The two normalization factors $\mathcal{N}(z)$ and $\mathcal{S}(w)$
implement cross-channel gain control within the Fisher information
framework.  The konio cross-normalization factor
\begin{equation}\label{eq:normk}
  \mathcal{N}(z) = \frac{\tau}{\tau + z^2}
\end{equation}
suppresses koniocellular sensitivity at chromaticities with high L-M
opponent activity ($|z| \gg 0$).  This is a form of
Carandini--Heeger divisive normalization~\cite{carandini12} applied
to the chromatic domain: the normalization pool consists of L-M
opponent activity in the parvocellular surround, which suppresses the
S-cone pathway response.  Solomon, Peirce \&
Lennie~\cite{solomon04} measured precisely this type of
cross-mechanism suppression in macaque V1, finding that L-M gratings
reduce the responses of S-cone-driven neurons.
\\
\\
At the adaptation white ($z \approx 0$), $\mathcal{N} \approx 1$
and the konio channel operates at full sensitivity.  At
spectrum-locus chromaticities with strong L-M contrast
($z^2 \gg \tau$), $\mathcal{N} \to 0$ and the konio channel is
suppressed.
\\
\\
The S-cone disinhibition factor
\begin{equation}\label{eq:scone_disinh}
  \mathcal{S}(w) = 1 + \frac{\delta_{\mathrm{nc}}\,\sigma_w^2}
  {w^2 + \sigma_w^2}
\end{equation}
modulates the non-cardinal population in the opposite sense: it
boosts non-cardinal sensitivity when S-cone input is absent.
The neurophysiological basis and geometric consequences of this factor
are described in Section~\ref{sec:scone_disinh}.
\\
\\
The non-cardinal preferred direction varies across the chromaticity
plane as a linear function of the opponent coordinates:
\begin{equation}\label{eq:phi}
  \theta(z,w) = \theta_0 + \theta_1\, z
  + \theta_2\,(w - w_0),
  \qquad
  w_0 = \frac{(0.5)^{p_s}}{(0.5)^{p_s} + \kappa_s}.
\end{equation}
Here $w_0$ is the koniocellular coordinate at the adaptation white
($r = 0.5$ when all adapted cone signals equal unity).  The linear
field is the simplest parameterization that allows ellipse
orientations to vary across the chromaticity diagram.

\begin{proposition}[Positive definiteness]\label{prop:pd}
$G(z,w;\,Y)$ as defined in~\eqref{eq:G} is positive definite for all
$Y > 0$, $\gamma > 0$, $\tau > 0$, $\delta_{\mathrm{nc}} \geq 0$,
$\sigma_w > 0$, and $\eta \geq 0$.
\end{proposition}

\begin{proof}
Since $\tau > 0$ and $z^2 \geq 0$, we have
$0 < \mathcal{N} \leq 1$.  The diagonal matrix
$\diag(Y^{a_P}, \mathcal{N}\,\gamma^2 Y^{a_K})$ therefore has
strictly positive entries for $Y > 0$ and $\gamma > 0$, hence is
positive definite.  Since $\delta_{\mathrm{nc}} \geq 0$ and
$\sigma_w > 0$, we have $\mathcal{S}(w) \geq 1 > 0$.
The rank-1 matrix
$\eta^2\,\mathcal{S}(w)\, Y^{(a_P+a_K)/2}\,\nhat\,\nhat^T$ is
positive semi-definite (its eigenvalues are $0$ and
$\eta^2\,\mathcal{S}\, Y^{(a_P+a_K)/2} \geq 0$).  The sum of a
positive definite matrix and a positive semi-definite matrix is
positive definite.
\end{proof}

\subsection{S-cone disinhibition of the non-cardinal population}
\label{sec:scone_disinh}

The non-cardinal V1 population receives convergent input from both
parvocellular and koniocellular pathways.  Solomon, Peirce \&
Lennie~\cite{solomon04} measured cross-mechanism suppression in
macaque V1 and found that S-cone signals modulate the responses of
non-cardinal chromatic neurons: S-cone activity suppresses the
sensitivity of neurons tuned to intermediate chromatic directions.
\\
\\
When S-cone input is absent, this suppression is released.  On the
spectrum locus above approximately 510~nm, the S-cone fundamental
sensitivity is negligible, so the koniocellular opponent coordinate
$w$ approaches zero.  
\\
\\
The S-cone disinhibition factor~\eqref{eq:scone_disinh}
resolves this tradeoff.  When $w \to 0$ (no S-cone signal),
$\mathcal{S} \approx 1 + \delta_{\mathrm{nc}}$, boosting
non-cardinal sensitivity by a factor of $1 + \delta_{\mathrm{nc}}$.
At the adaptation white ($w = w_0 \approx 6\sigma_w$ for the
fitted parameters), $\mathcal{S} \approx 1.27$, a modest residual
boost.  At chromaticities with high S-cone content
($w \gg w_0$), $\mathcal{S} \to 1$.  The ratio of the
disinhibition at S-cone-absent wavelengths to its value at the
adaptation white is approximately
$(1 + \delta_{\mathrm{nc}})/\mathcal{S}(w_0) \approx 9$,
providing strong differential sensitivity.
\\
\\
The important geometric property is that the koniocellular coordinate $w$
separates the Wright problem wavelengths from the MacAdam
centres.  At Wright 520--560~nm, $w/w_0$ ranges from approximately
0.03 to 0.23 while at all 25 MacAdam centres,
$w/w_0$ is substantially larger because broadband stimuli inside the
gamut retain measurable S-cone excitation even at red and green
chromaticities.  This separation is a
consequence of the spectrum locus having negligible S-cone content
above 510~nm, while all broadband stimuli retain measurable S-cone excitation.  No other
opponent coordinate provides comparable separation.

\section{The Metric}\label{sec:metric}

\subsection{Pullback to chromaticity coordinates}

The opponent coordinate map $\Phi\colon (x,y) \mapsto (z,w)$ at fixed
luminance $Y$ has a $2\times 2$ Jacobian
\begin{equation}\label{eq:jac}
  J_{\mathrm{opp}}
  = \frac{\partial(z,w)}{\partial(x,y)},
\end{equation}
computed numerically by central differences.  The Riemannian metric on
the chromaticity plane, expressed in CIE $(x,y)$ coordinates, is the
pullback:
\begin{equation}\label{eq:gchrom}
  g_{\mathrm{chrom}}(x,y;\,Y) = J_{\mathrm{opp}}^T\, G(z,w;\,Y)\,
  J_{\mathrm{opp}}.
\end{equation}

\subsection{Luminance metric: Magnocellulur pathway}
 
Luminance discrimination is mediated primarily by the magnocellular
pathway, which receives additive $L+M$ input and has high contrast
sensitivity~\cite{kaplan86,merigan93}.  We derive the luminance
metric from the same Poisson Fisher information framework used for
the chromatic channels.
\\
\\
A magnocellular neuron encoding luminance has a firing rate that
depends on $Y$ through a power law $f(Y) = c\,Y^{a_M}$, where
$a_M$ is the magnocellular luminance-response exponent and $c$
absorbs the neuron's baseline sensitivity.  The one-dimensional
Poisson Fisher information for the parameter $Y$ is
\[
  g_{YY}^{(\text{single})}
  = \frac{[f'(Y)]^2}{f(Y)}
  = \frac{[c\,a_M\,Y^{a_M - 1}]^2}{c\,Y^{a_M}}
  = c\,a_M^2\;Y^{a_M - 2}.
\]
Summing over a population of magnocellular neurons (as in
Section~\ref{sec:lum_scaling}) gives an aggregate Fisher information
\begin{equation}\label{eq:gYY_raw}
  g_{YY}(Y) = \lambda_M\;Y^{a_M - 2},
\end{equation}
where $\lambda_M$ absorbs the population size and per-neuron
parameters.
\\
\\
It is convenient to reparameterize by defining
$\beta_Y = (2 - a_M)/2$, so that $a_M = 2(1 - \beta_Y)$ and
$a_M - 2 = -2\beta_Y$.  Then~\eqref{eq:gYY_raw} becomes
\begin{equation}\label{eq:gYY}
  g_{YY}(Y) = \psi_0^2 \,/\, Y^{2\beta_Y},
\end{equation}
where $\psi_0^2 = \lambda_M$ absorbs the population-level prefactor.
The reparameterization is natural because $\beta_Y$ directly governs
the power law of luminance discrimination thresholds: since
$\Delta Y \propto 1/\sqrt{g_{YY}} \propto Y^{\beta_Y}$, a larger
$\beta_Y$ means thresholds grow faster with luminance.
\\
\\
The two classical limiting cases correspond to specific values of the
magnocellular response exponent.  When $a_M = 1$ (linear response),
$\beta_Y = 0.5$ and $\Delta Y \propto Y^{0.5}$, which is the
de~Vries--Rose law characteristic of photon-noise-limited detection.
As $a_M \to 0$ (nearly constant response), $\beta_Y \to 1$ and
$\Delta Y \propto Y$, approaching Weber's law for luminance
($\Delta Y / Y \approx \text{const}$), established since K\"onig \&
Brodhun~\cite{konig89}.  The fitted value $\beta_Y = 0.84$
corresponds to $a_M = 0.32$, a mildly sublinear magnocellular
response: luminance discrimination is closer to Weber's law than to
de~Vries--Rose, consistent with the psychophysical literature at
moderate photopic luminances~\cite{hood86}.  Note that exact Weber's
law ($\beta_Y = 1$ exactly) would require $a_M = 0$, making $g_{YY}$
identically zero in the Poisson model.  The fact that the fitted
$\beta_Y$ falls short of unity is therefore not merely an empirical
observation but a structural prediction of the Poisson framework: the
magnocellular pathway must retain some luminance dependence in its
firing rates for the Fisher information to be nonzero.

\subsection{The full three-dimensional metric}

The chromaticity and luminance metrics are assumed to be block diagonal:
\begin{equation}\label{eq:g3d}
  g(x,y,Y) =
  \begin{pmatrix}
    g_{\mathrm{chrom}}(x,y;\,Y) & \mathbf{0} \\
    \mathbf{0}^T & g_{YY}(Y)
  \end{pmatrix}.
\end{equation}
This corresponds to the assumption that the chromatic (parvo/konio)
and luminance (magno) pathways contribute independently to the Fisher
information.  Krauskopf, Williams \&
Heeley~\cite{krauskopf82} showed that the cardinal directions of
color space (L-M, S, and L+M) are approximately independent
mechanisms.  
\\
\\
The Koenderink et~al.~\cite{koenderink26} data provide a quantitative
test of this independence.  Each Koenderink ellipsoid is a
$3 \times 3$ covariance matrix $\Sigma$ in sRGB coordinates.  To
assess chromaticity-luminance coupling, we rotate each $\Sigma$ into
a basis where one axis is the achromatic direction
$\mathbf{a} = (1, 1, 1)/\sqrt{3}$ (equal increments in R, G, B
correspond to a pure luminance change) and the other two axes span
the chromatic plane perpendicular to $\mathbf{a}$.  In this basis,
the $3 \times 3$ covariance decomposes into a $2 \times 2$
chromatic block, a $1 \times 1$ luminance entry, and $2 \times 1$
cross terms.  If the block-diagonal assumption holds exactly, the
cross terms vanish.  In practice, the ratio of the cross-term
magnitudes to the geometric mean of the corresponding diagonal
entries ranges from $0.01$ to $0.22$ at the five achromatic
reference colours confirming
that the coupling is small but nonzero.

\section{Parameters}\label{sec:params}

The model has 17 free parameters, each with a neurophysiological origin:
\\[6pt]
\begin{tabular}{clll}
\toprule
\# & Symbol & Neurophysiological origin & Fitted value \\
\midrule
0 & $p_s$ & S-cone transduction exponent & 0.597 \\
1 & $\kappa_s$ & S-cone half-saturation & 8.616 \\
2 & $\gamma$ & Konio/parvo sensitivity ratio & 12.85 \\
3 & $\eta$ & Non-cardinal V1 sensitivity & 2.355 \\
4 & $\theta_0$ & Mean non-cardinal direction & $-2.244$ \\
5 & $\theta_1$ & Direction gradient along $z$ & 0.994 \\
6 & $\theta_2$ & Direction gradient along $w$ & 1.508 \\
7 & $z_{\max}$ & L-M saturation ceiling & 1.993 \\
8 & $p_z$ & L-M Naka--Rushton exponent & 1.022 \\
9 & $\kappa_z$ & L-M half-saturation & 1.393 \\
10 & $a_P$ & Parvo luminance exponent & 0.717 \\
11 & $a_K$ & Konio luminance exponent & 0.568 \\
12 & $\psi_0$ & Magnocellular Weber fraction & 0.738 \\
13 & $\beta_Y$ & Magno luminance exponent & 0.841 \\
14 & $\tau$ & Konio cross-normalization & 0.796 \\
15 & $\delta_{\mathrm{nc}}$ & S-cone disinhibition strength & 10.00 \\
16 & $\sigma_w$ & Disinhibition half-saturation & 0.012 \\
\bottomrule
\end{tabular}
\\[6pt]
The fitted values are from joint optimization across four
datasets (Section~\ref{sec:results}).  Parameters 0--9 govern the
chromaticity metric.  Parameters 10--11 control how the chromatic
metric scales with luminance.  Parameters 12--13 are constrained only
by three-dimensional data.  Parameter~14 controls the divisive
normalization of the koniocellular channel by parvocellular activity.
Parameters~15--16 control the S-cone disinhibition of the
non-cardinal population.

\section{Computing the Discrimination Ellipsoid}\label{sec:recipe}

This section collects all the pieces into a single recipe.  Given a
color stimulus with CIE coordinates $(x, y, Y)$ and an adaptation
white with cone signals $(L_W, M_W, S_W)$, the following steps
produce the predicted $3 \times 3$ discrimination ellipsoid in any
desired coordinate system.
\\
\\
\emph{Step 1. Cone signals and adaptation.}
Compute the tristimulus values $X = Yx/y$, $Z = Y(1{-}x{-}y)/y$,
then the cone signals $(L, M, S) = M_{\mathrm{SP}}\,(X, Y, Z)^T$
using~\eqref{eq:cone}.  Divide by the adaptation white:
$L_a = L/L_W$, $M_a = M/M_W$, $S_a = S/S_W$.
\\
\\
\emph{Step 2. Opponent coordinates.}
Compute the two chromaticity coordinates.  For the L-M channel,
apply the saturating Naka--Rushton to the log cone ratio:
\[
  u = \log(L_a / M_a), \qquad
  z = z_{\max}\;\mathrm{sign}(u)\;
  \frac{|u|^{p_z}}{|u|^{p_z} + \kappa_z}.
\]
For the S channel:
\[
  r = \frac{S_a}{L_a + M_a},
  \qquad
  w = \frac{r^{p_s}}{r^{p_s} + \kappa_s}.
\]
These depend only on chromaticity, not on luminance
(Proposition~\ref{prop:invariance}).
\\
\\
\emph{Step 3. Non-cardinal direction.}
Compute the preferred direction of the non-cardinal V1 population:
\[
  \theta = \theta_0 + \theta_1\, z + \theta_2\,(w - w_0),
  \qquad
  w_0 = \frac{(0.5)^{p_s}}{(0.5)^{p_s} + \kappa_s},
  \qquad
  \nhat = \begin{pmatrix} \cos\theta \\ \sin\theta \end{pmatrix}.
\]
\\
\\
\emph{Step 4. Channel metric.}
Compute the konio cross-channel normalization factor:
\[
  \mathcal{N} = \frac{\tau}{\tau + z^2}.
\]
Compute the S-cone disinhibition factor:
\[
  \mathcal{S} = 1 + \frac{\delta_{\mathrm{nc}}\,\sigma_w^2}
  {w^2 + \sigma_w^2}.
\]
Assemble the $2 \times 2$ Fisher information matrix in opponent
coordinates:
\[
  G = G_P + G_K + G_N =
  \begin{pmatrix}
    Y^{a_P} & 0 \\
    0 & \mathcal{N}\,\gamma^2\, Y^{a_K}
  \end{pmatrix}
  + \eta^2\,\mathcal{S}\, Y^{(a_P + a_K)/2}\, \nhat\,\nhat^T.
\]
The diagonal terms represent the parvocellular and koniocellular
populations, with konio sensitivity reduced by $\mathcal{N}$ at
chromaticities with high L-M opponent activity.  The rank-1 term
represents the non-cardinal population, with sensitivity modulated
by $\mathcal{S}$: strongly boosted when the S-cone pathway is inactive
($w \approx 0$) and only modestly boosted when S-cone input is
present ($w \gg \sigma_w$).
\\
\\
\emph{Step 5. Jacobian.}
Compute the $2 \times 2$ Jacobian $J_{\mathrm{opp}} =
\partial(z,w)/\partial(x,y)$ by repeating Steps~1--2 at four
neighboring chromaticities and using central differences:
\[
  [J_{\mathrm{opp}}]_{i1}
  = \frac{\Phi_i(x{+}h,\,y) - \Phi_i(x{-}h,\,y)}{2h},
  \qquad
  [J_{\mathrm{opp}}]_{i2}
  = \frac{\Phi_i(x,\,y{+}h) - \Phi_i(x,\,y{-}h)}{2h},
\]
where $\Phi = (z, w)$ and $h \approx 10^{-7}$.
\\
\\
\emph{Step 6. Chromaticity metric.}
Pull the channel metric back to CIE coordinates:
\[
  g_{\mathrm{chrom}} = J_{\mathrm{opp}}^T\, G\, J_{\mathrm{opp}}.
\]
This is a $2 \times 2$ positive definite matrix whose inverse gives
the discrimination ellipse at luminance $Y$.
\\
\\
\emph{Step 7. Luminance metric.}
Compute the scalar luminance metric:
\[
  g_{YY} = \psi_0^2 \,/\, Y^{2\beta_Y}.
\]
\\
\\
\emph{Step 8. Full 3D metric.}
Assemble the block-diagonal $3 \times 3$ metric in $(x, y, Y)$
coordinates:
\[
  g =
  \begin{pmatrix}
    g_{\mathrm{chrom}} & \mathbf{0} \\
    \mathbf{0}^T & g_{YY}
  \end{pmatrix}.
\]
\\
\\
\emph{Step 9 (if needed). Coordinate transform.}
If the data are in display coordinates $(R_\gamma, G_\gamma,
B_\gamma)$ rather than $(x, y, Y)$, compute the $3 \times 3$ Jacobian
$J_{\mathrm{full}} = \partial(x,y,Y)/\partial(R_\gamma, G_\gamma,
B_\gamma)$ by central differences on the display model, then
transform:
\[
  g_{\mathrm{RGB}} = J_{\mathrm{full}}^T\, g\, J_{\mathrm{full}}.
\]
\\
\\
\emph{Step 10. Output.}   The predicted discrimination ellipsoid is
$\Sigma = g_{\mathrm{RGB}}^{-1}$ (or $g^{-1}$ if working in
$(x,y,Y)$), up to the criterion constant $c^2$ that is absorbed
by the STRESS scale factor (Section~\ref{sec:fisher}).
The eigenvectors of $\Sigma$ give the principal axes of
the ellipsoid; the square roots of the eigenvalues give the
semi-axis lengths.  For a 2D chromaticity ellipse at fixed $Y$, use
$g_{\mathrm{chrom}}^{-1}$ from Step~6.
\\
\\
The entire computation requires only matrix multiplication, one
logarithm, one power function, and numerical differentiation.  No
iterative solvers or special functions are needed.

\section{Empirical Validation}\label{sec:results}

\subsection{Datasets}

We test the model against four datasets that span threshold
discrimination at varied chromaticities, luminances, and stimulus
conditions:

\begin{enumerate}[label=(\roman*)]
\item \emph{MacAdam (1942)}~\cite{macadam42}: 25 threshold
  discrimination ellipses, single observer (PGN), constant luminance
  $Y = 48$~cd/m$^2$, Illuminant~C adaptation.  Data from Wyszecki \&
  Stiles~\cite{wyszecki00}.
\item \emph{Koenderink et~al.\ (2026)}~\cite{koenderink26}: 35
  three-dimensional ellipsoids (Notable Qualitative Differences), 8
  observers (group median), $Y = 7.1$--$129.2$~cd/m$^2$, D65
  adaptation.
\item \emph{Wright (1941)}~\cite{wright41}: 19 wavelength
  discrimination thresholds ($\Delta\lambda$ at threshold) along the
  spectrum locus from 440--620~nm, single observer, $Y = 48$~cd/m$^2$,
  Illuminant~C adaptation.
\item \emph{Huang et~al.\ (2012)}~\cite{huang15}: 17 chromaticity
  discrimination ellipses at the CIE-recommended colour centres,
  16 observers, pass/fail threshold method, printed samples under D65,
  CIE 1964 10$^\circ$ observer.  The centres span luminances
  $Y = 9.4$--$69.0$~cd/m$^2$ and chromaticities from near-neutral
  to high chroma.
\end{enumerate}

\subsection{Coordinate transformations}\label{sec:transforms}

The model produces a $2\times 2$ chromaticity metric
$g_{\mathrm{chrom}}(x,y;\,Y)$ in CIE~$(x,y)$ coordinates and a
$3\times 3$ metric $g(x,y,Y)$ in CIE~$(x,y,Y)$ coordinates.  Each
dataset reports discrimination data in a different coordinate system.
For each dataset we choose whichever transformation direction avoids
numerical instability: the observed data are transformed to CIE~$xy$
where possible, and the model is transformed to the data's native
coordinates where necessary.  In all cases the predicted and observed
quantities are compared in the same coordinate system.  The
transformations are as follows.
\\
\\
\emph{MacAdam.}  The data are already in CIE~$(x,y)$ coordinates.
Each ellipse is specified by semi-axes $(a,b)$ and orientation
$\vartheta$, which define an observed $2\times 2$ metric tensor
directly:
\[
  g_{\mathrm{obs}} = R(\vartheta)\,
  \diag(1/a^2,\;1/b^2)\, R(\vartheta)^T,
\]
where $R(\vartheta)$ is a rotation matrix.  No coordinate
transformation is needed; the model's $g_{\mathrm{chrom}}$ is
compared directly with $g_{\mathrm{obs}}$ at each ellipse centre.
\\
\\
\emph{Koenderink.}  The data are $3\times 3$ covariance matrices
$\Sigma_{\mathrm{obs}}$ in gamma-encoded sRGB coordinates.  The
comparison is performed in sRGB: the model's metric is transformed
to sRGB via the $3\times 3$ Jacobian
\[
  J_{\mathrm{full}} =
  \frac{\partial(x,\,y,\,Y)}%
       {\partial(R_\gamma,\, G_\gamma,\, B_\gamma)}
\]
computed by central differences ($h = 10^{-7}$) through the sRGB
display model (gamma decoding followed by the $3\times 3$
sRGB-to-XYZ matrix at the display's peak luminance
$Y_{\mathrm{white}} = 214$~cd/m$^2$).  The predicted covariance in
sRGB is then
\[
  \Sigma_{\mathrm{pred}} =
  \bigl(J_{\mathrm{full}}^T\, g\, J_{\mathrm{full}}\bigr)^{-1}.
\]
Equivalently, one could transform $\Sigma_{\mathrm{obs}}$ to
$(x,y,Y)$ via $J_{\mathrm{full}}\,\Sigma_{\mathrm{obs}}\,
J_{\mathrm{full}}^T$ and compare with $g^{-1}$; the two approaches
are related by a change of basis and yield the same STRESS.
Predicted and observed covariances are compared along 26 probe
directions.
\\
\\
\emph{Wright.}  The data are scalar wavelength thresholds
$\Delta\lambda$ at 19 points along the spectrum locus.  The
chromaticity coordinates $(x(\lambda), y(\lambda))$ and their
gradients $\mathbf{v}(\lambda) = (\partial x/\partial\lambda,\;
\partial y/\partial\lambda)$ were computed from the CIE~1931
$2^\circ$ colour matching functions at 0.5~nm spacing using central
differences.  The model predicts a threshold proportional to the
inverse of the metric's sensitivity along the wavelength direction:
\[
  \Delta\lambda_{\mathrm{pred}} \;\propto\;
  \frac{1}{\sqrt{\mathbf{v}^T\, g_{\mathrm{chrom}}\, \mathbf{v}}}.
\]
The proportionality constant (which depends on the observer's
criterion) is absorbed by the optimal scaling factor $F$ in STRESS.
\\
\\
\emph{Huang.}  The data are ellipses in the CIELAB
$a^*_{10}\,b^*_{10}$ plane (10$^\circ$ observer, D65 white).  Each
ellipse is specified by a semi-major axis $A$, an axis ratio $A/B$,
and an orientation $\vartheta$.  These define an observed metric
tensor in $(a^*,b^*)$ coordinates:
\[
  g_{a^*b^*} = R(\vartheta)\,
  \diag(1/A^2,\;1/B^2)\, R(\vartheta)^T.
\]
To transform into CIE~$(x,y)$ we compute the $2\times 2$ Jacobian
$J_{ab} = \partial(a^*_{10},b^*_{10})/\partial(x,y)$ numerically
($h = 10^{-5}$) by perturbing $(x,y)$ at fixed $Y$ and passing
through the CIELAB forward equations with the 10$^\circ$ D65 white
point $(X_n, Y_n, Z_n) = (94.811, 100.0, 107.304)$.  The observed
metric in~$(x,y)$ coordinates is then
\[
  g_{\mathrm{obs}} = J_{ab}^T\, g_{a^*b^*}\, J_{ab},
\]
which is compared with $g_{\mathrm{chrom}}$ evaluated at the
colour centre.  Note that the model uses $2^\circ$ cone fundamentals
throughout which creates small errors in our predictions for this dataset

\subsection{Evaluation measures}

\emph{STRESS} (Standardized Residual Sum of Squares) compares
predicted and observed ellipse radii at multiple probe directions,
with an optimal global scaling factor $F$:
\begin{equation}\label{eq:stress}
  \mathrm{STRESS} = 100\,\sqrt{\frac{\sum_i (r_i - F p_i)^2}
  {\sum_i r_i^2}}
\end{equation}
where $r_i$ and $p_i$ are observed and predicted radii in probe
direction~$i$.

\subsection{Results}

\begin{table}[ht!]
\centering
\label{tab:results}
\begin{tabular}{llcccc}
\toprule
Dataset & Model  & STRESS  \\
\midrule
MacAdam & CIELAB  & 41.9 \\
MacAdam & CIEDE2000  & 37.6  \\
MacAdam & CIECAM02-UCS  & 30.2 \\
MacAdam & CAM16-UCS  & 30.5 \\
MacAdam & Fisher Information   & 23.9  \\
\midrule
Koenderink et al. & CIELAB  & 26.2  \\
Koenderink et al. & CIEDE2000  & 25.2  \\
Koenderink et al. & CIECAM02-UCS  & 28.1  \\
Koenderink et al. & CAM16-UCS  & 29.2  \\
Koenderink et al. & Fisher Information  & 20.8  \\
\midrule
Wright & CIELAB  & 61.3  \\
Wright & CIEDE2000  & 59.3  \\
Wright & CIECAM02-UCS  & 58.8  \\
Wright & CAM16-UCS  & 62.6  \\
Wright & Fisher Information  & 30.1  \\
\midrule
Huang et al. & CIELAB  & 35.8  \\
Huang et al. & CIEDE2000  & 30.7  \\
Huang et al. & CIECAM02-UCS  & 28.8 \\
Huang et al. & CAM16-UCS  & 26.1  \\
Huang et al. & Fisher Information  & 30.8  \\
\bottomrule
\end{tabular}
\caption{Performance across four datasets. Fisher Information is the 17-parameter
model optimized simultaneously across all four datasets with weights
$1.5\times$ MacAdam, $1.0\times$ Koenderink, $1.0\times$ Wright,
$0.75\times$ Huang.}
\end{table}
\noindent
The Fisher Information model, with a single set of 17 parameters fitted
jointly to four datasets, is the only model that achieves competitive
performance across all datasets simultaneously.
\\
\\
On MacAdam it achieves STRESS of 23.9 versus 30.2 for
the best standard model (CIECAM02-UCS~\cite{luo06}), a 21\%
improvement (Figure~\ref{fig:macadam}).  On Koenderink it achieves
STRESS of 20.8 (Figure~\ref{fig:koenderink}).
\\
\\
On Wright wavelength discrimination, the Fisher Information model achieves
STRESS of 30.1 versus 58.8 for the best standard model
(CIECAM02-UCS), a 49\% improvement.  All standard metrics cluster
at STRESS 58--63 on this dataset because none possess a mechanism
for the large sensitivity variations along the spectrum locus.
Two normalization mechanisms contribute to this improvement.  The
konio cross-normalization factor $\mathcal{N}$ suppresses the
S-cone pathway at chromaticities with high L-M opponent activity:
at blue-green wavelengths around 470--500~nm, the substantial
L-M opponent contrast reduces $\mathcal{N}$ well below unity,
suppressing the excess konio sensitivity that
otherwise predicts thresholds well below Wright's observations.
The S-cone disinhibition factor $\mathcal{S}$ boosts non-cardinal
sensitivity at chromaticities where the S-cone pathway is
inactive.
\\
\\
On the Huang et al.\ threshold ellipses (Figure~\ref{fig:huang}), the model
achieves STRESS of 30.8, comparable to CIEDE2000 (30.7) and better
than CIELAB (35.8), though below CAM16-UCS (26.1) on this
individual dataset.
CAM16-UCS has an advantage on Huang,
likely because its chromatic adaptation transform and chroma
compression were tuned to suprathreshold data at varied chromaticities
that partially overlap with the Huang conditions.  However, CAM16-UCS
fails on Wright (62.6), demonstrating that no standard metric
achieves the cross-dataset balance of the Fisher Information model.

\begin{figure}[H]
\centering
\includegraphics[width=0.85\textwidth]{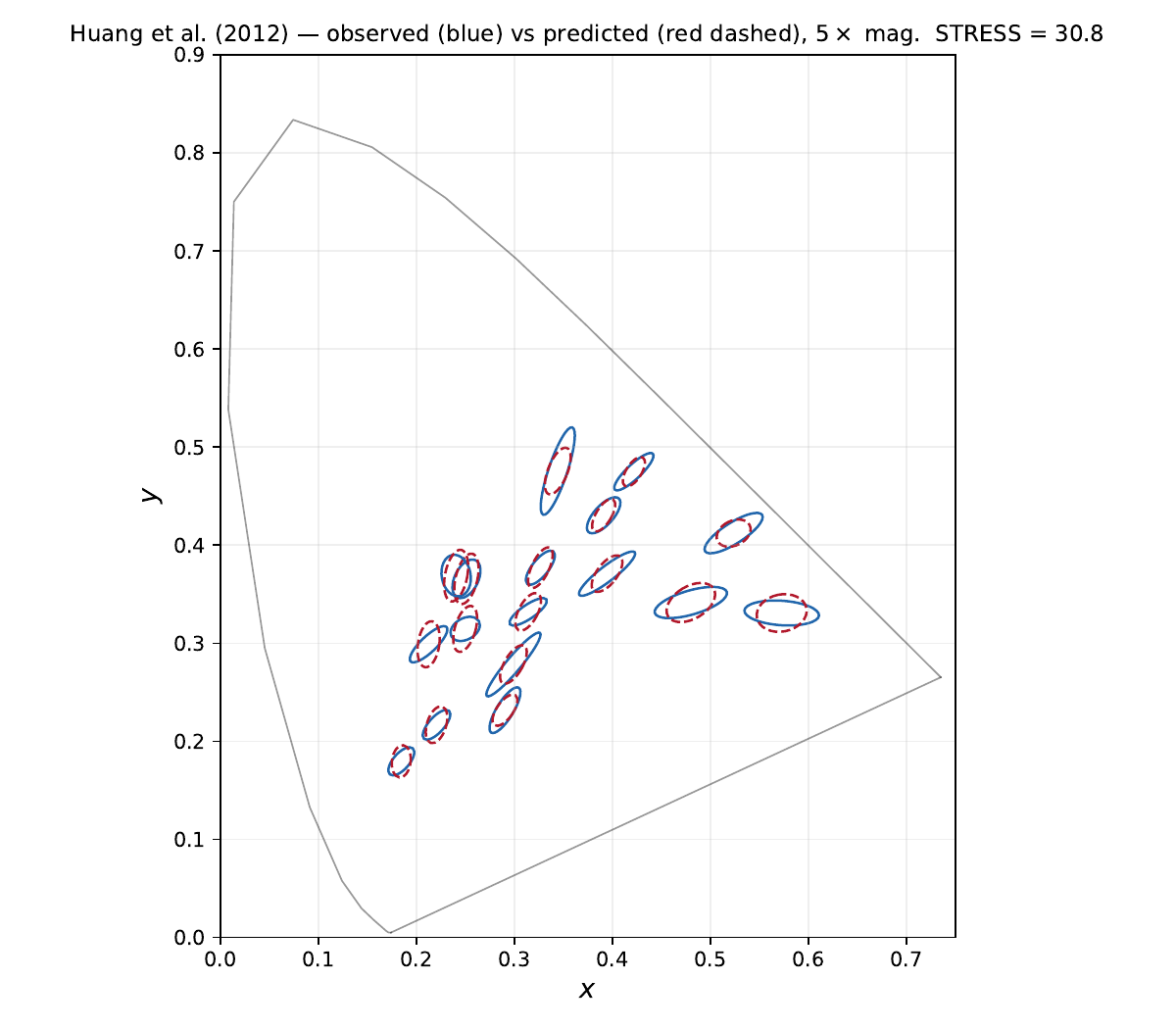}
\caption{Huang et~al.\ (2012) threshold ellipses: observed (blue,
solid) versus predicted (red, dashed) at $5\times$ magnification in
CIE~$(x,y)$ coordinates, after transforming the original CIELAB
ellipses.}
\label{fig:huang}
\end{figure}

\begin{figure}[H]
\centering
\includegraphics[width=\textwidth]{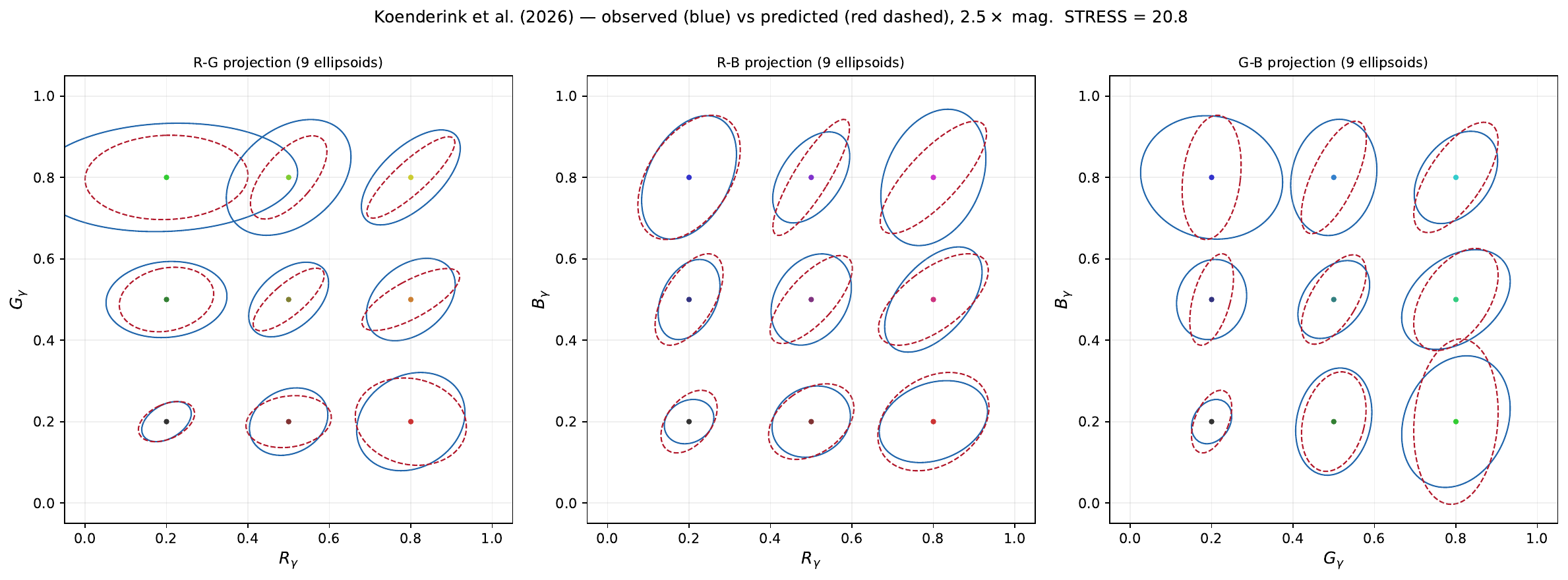}
\caption{Koenderink et~al.\ (2026) discrimination ellipsoids
projected onto three pairs of sRGB coordinates: observed (blue,
solid) versus predicted (red, dashed) at $2.5\times$ magnification.
A mostly non-overlapping subset of ellipsoids is shown in each panel (number
indicated in subtitle); the STRESS-optimal scale factor is computed
from all 35 ellipsoids.  Centre dots are coloured by sRGB value.
STRESS = 20.8.}
\label{fig:koenderink}
\end{figure}

\begin{figure}[H]
\centering
\includegraphics[width=0.85\textwidth]{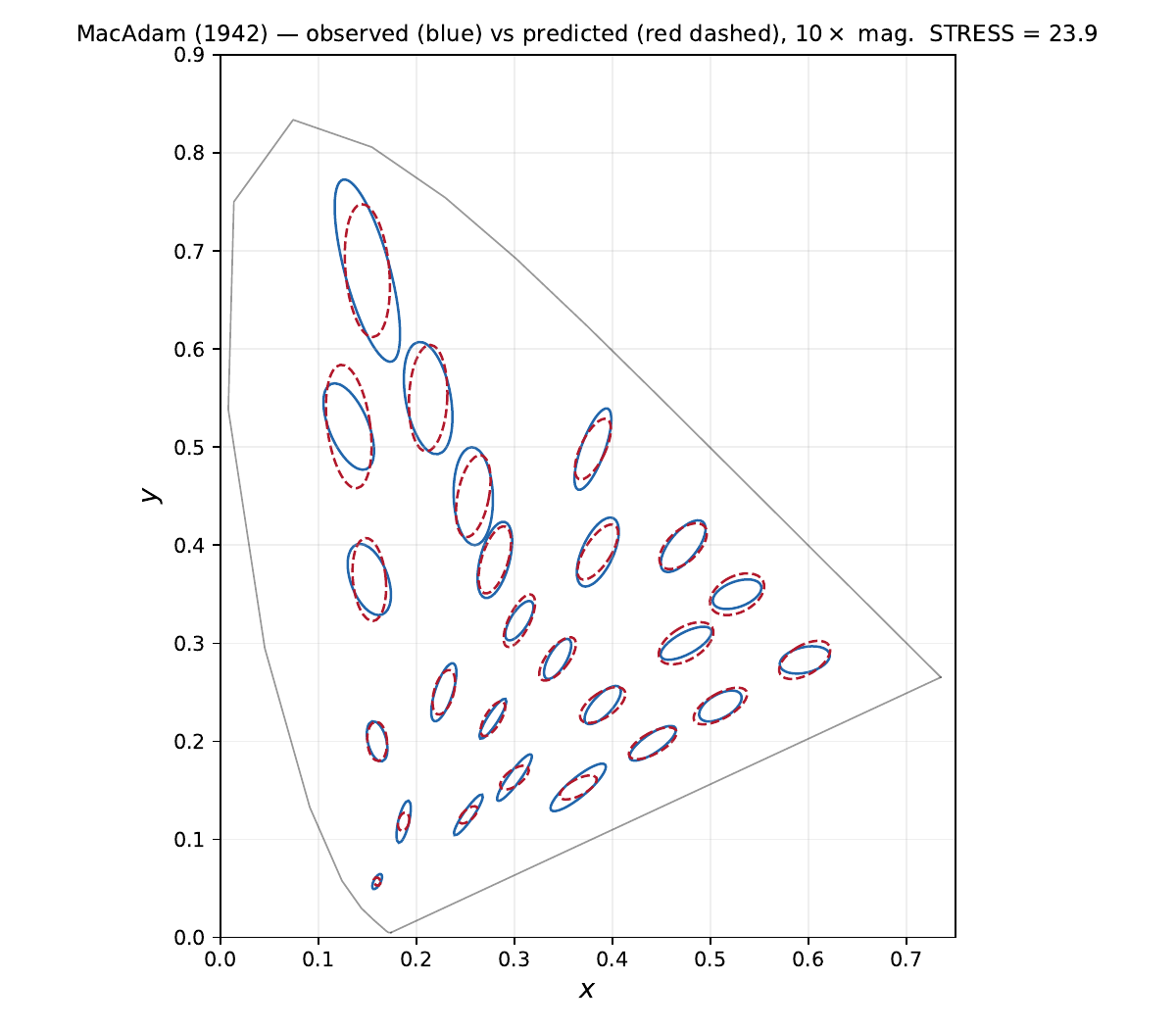}
\caption{MacAdam (1942) discrimination ellipses: observed (blue, solid)
versus predicted (red, dashed) at $10\times$ magnification.  The
predicted ellipses use the STRESS-optimal global scale factor $F$.
The model captures ellipse orientations and aspect ratios across
the chromaticity diagram.
STRESS = 23.9.}
\label{fig:macadam}
\end{figure}

\end{document}